\documentclass[orivec,runningheads,10pt]{llncs}
\usepackage{tikz}
\usepackage{pgfplots}
\usepackage{amsmath}
\usepackage{amsfonts}
\usepackage{enumerate}
\usepackage{comment}
\usepackage{amssymb}
\usepackage{graphicx}
\usepackage{subfigure}
\usepackage{cite}

\newtheorem{ass}{Assumption}

\newcommand{\bc}{\ensuremath{\mathbf{c}}}
\newcommand{\be}{\ensuremath{\mathbf{e}}}

\newcommand{\bm}{\ensuremath{\mathbf{m}}}
\newcommand{\ba}{\ensuremath{\mathbf{a}}}
\newcommand{\bb}{\ensuremath{\mathbf{b}}}

\newcommand{\bh}{\ensuremath{\mathbf{h}}}
\newcommand{\bs}{\ensuremath{\mathbf{s}}}

\newcommand{\bA}{\ensuremath{\mathbf{A}}}
\newcommand{\bG}{\ensuremath{\mathbf{G}}}

\newcommand{\bK}{\ensuremath{\mathbf{K}}}

\newcommand{\Ftwo}{\ensuremath{\mathbb{F}_2}}
\newcommand{\bH}{\ensuremath{\mathbf{H}}}

\newcommand{\weight}[1]{\ensuremath{\mathrm{w}_{\mathrm{H}}\left(#1\right)}}

\newcommand{\pr}[1]{\ensuremath{\text{\tt Prob}\left(#1\right)}}
\newcommand{\prJ}[1]{\text{\tt Prob}_{\mathcal{P}_n^*}\left(#1\right)}
\newcommand{\card}[1]{\left|#1\right|}
\newcommand{\supp}[1]{\mathrm{S}\left(#1\right)}

\newcommand{\tleft}[1]{\hat{t}^{(#1)}}
\newcommand{\val}{\ensuremath{\mathbf{y}}}

\newcommand{\upc}{\ensuremath{\mathtt{upc}}}

\newcommand{\Pfzero}[1]{\ensuremath{\mathtt{P}_{f \mid 0}(#1)}}
\newcommand{\Pfone}[1]{\ensuremath{\mathrm{P}_{f \mid 1}(#1)}}
\newcommand{\Pmzero}[1]{\ensuremath{\mathrm{P}_{m \mid 0}(#1)}}
\newcommand{\Pmone}[1]{\ensuremath{\mathrm{P}_{m \mid 1}(#1)}}

\newcommand{\pErrOneUnsat}{\ensuremath{\rho_{1,\mathtt{u}}}}

\newcommand{\pErrZeroUnsat}{\ensuremath{\rho_{0,\mathtt{u}}}}

\newcommand{\Pf}[1]{\Pfone{#1}}
\newcommand{\Pu}[1]{\Pmzero{#1}}

\newcommand{\cPu}[1]{\Pfzero{#1}}

\newcommand{\pic}{\pErrOneUnsat}

\newcommand{\pci}{\pErrZeroUnsat}

\usepackage[acronym,shortcuts]{glossaries}
\glsdisablehyper

\newacronym{ML}{ML}{Maximum Likelihood}
\newacronym{SDP}{SDP}{Syndrome Decoding Problem}
\newacronym{ISD}{ISD}{Information Set Decoding}
\newacronym{QC-MDPC}{QC-MDPC}{Quasi-Cyclic Moderate-Density Parity-Check}
\newacronym{QC-LDPC}{QC-LDPC}{Quasi-Cyclic Low-Density Parity-Check}
\newacronym{LDPC}{LDPC}{Low-Density Parity-Check}
\newacronym{MDPC}{MDPC}{Moderate-Density Parity-Check}
\newacronym{QC}{QC}{Quasi-Cyclic}
\newacronym{DFR}{DFR}{Decoding Failure Rate}
\newacronym{BF}{BF}{Bit Flipping}
\newacronym{IND-CCA2}{IND-CCA2}{indistinguishability under adaptive chosen ciphertext attack}
\newacronym{RIP-BF}{RIP-BF}{Randomized In-Place Bit-Flipping}

\usepackage[noline,ruled,noend,linesnumbered]{algorithm2e}

\usepackage{breqn}
\usepackage{bbm}
\usepackage{latexsym}

\usepackage{tikz}
\usetikzlibrary{calc, chains,
                fit,
                positioning,
                shapes}

\usetikzlibrary{automata,arrows,positioning,calc}
\definecolor{myblue}{RGB}{0,0,205}
\definecolor{mygreen}{RGB}{80,160,80}
\definecolor{myred}{RGB}{178,34,34}
\definecolor{mygray}{RGB}{211,211,211}

\pgfplotsset{compat=1.14}
 \begin{document}
\title{A Code-specific Conservative Model for the Failure Rate of Bit-flipping 
Decoding of LDPC Codes with Cryptographic Applications}
\titlerunning{On the Failure Rate of Bit-flipping Decoding of LDPC Codes}
\author{Paolo Santini\inst{1,2}\and
Alessandro Barenghi \inst{3}\and
Gerardo Pelosi \inst{3} \and Marco Baldi\inst{1} \and Franco Chiaraluce\inst{1}}
\authorrunning{P. Santini et al.}
\institute{Universit\`a Politecnica delle Marche \and
Florida Atlantic University \and Politecnico di Milano\\
\email{p.santini@pm.univpm.it, \{alessandro.barenghi, gerardo.pelosi\}@polimi.it, 
\{m.baldi, f.chiaraluce\}@univpm.it}}
\maketitle
\begin{abstract}
Characterizing the decoding failure rate of iteratively decoded Low- and 
Moderate-Density Parity Check (LDPC/MDPC) codes is paramount to build 
cryptosystems based on them, able to achieve indistinguishability under adaptive 
chosen ciphertext attacks.
In this paper, we provide a statistical worst-case analysis of our proposed 
iterative decoder obtained through a simple modification of the classic in-place 
bit-flipping decoder.
This worst case analysis allows both to derive the worst-case behaviour 
of an LDPC/MDPC code picked among the family with the same length, rate and
number of parity checks, and a code-specific bound on the decoding failure rate.
The former result allows us to build a code-based cryptosystem enjoing the 
$\delta$-correctness property required by IND-CCA2 constructions, while
the latter result allows us to discard code instances which may have a 
decoding failure rate significantly different from the average one (i.e., 
representing weak keys), should they be picked during the key generation procedure.
\end{abstract}
\section{Introduction}
\label{sec:intro}
Code based cryptosystems, pioneered by McEliece~\cite{McE78}, are among the oldest 
public-key cryptosystems, and have survived a significant amount of cryptanalysis, 
remaining unbroken even for quantum-equipped adversaries~\cite{Ber10}.
This still holds true for both the original McEliece construction, and the one by
Niederreiter~\cite{nied}, when both instantiated with Goppa codes, as they both rely
on the same mathematical trapdoor, i.e., having the adversary solve the search 
version of the decoding problem for a general linear code, which was proven to 
be NP-Hard in~\cite{DBLP:journals/tit/BerlekampMT78}.

The public-key of such schemes corresponds to an obfuscated representation 
of the underlying error correcting code (either the generator matrix for McEliece, or the 
parity-check matrix for Niederreiter), equipped with a decoding technique that 
can efficiently correct a non-trivial amount of errors.
Since the obfuscated form of either the generator or the parity-check matrix
should be indistinguishable from the one of a random code with the same length and dimension, 
both the original McEliece and Niederreiter proposals have public-key sizes which
grow essentially quadratically in the error correction capacity of the code, 
on which the provided security level itself depends.

The large public-key size required in these cryptosystems has hindered 
their practical application in many scenarios. A concrete way of solving 
this problem is to employ codes described by matrices with a \ac{QC} structure, 
which result in public-key sizes growing linearly in the code length.
However, employing \ac{QC} algebraic codes has proven to be a security issue,
as the additional structure given by the quasi-cyclicity allows an attacker
to deduce the underlying structure of the secret code~\cite{faugere}.
By contrast, code families obtained from a random sparse parity-check matrix 
do not suffer from the same problem, and have lead to the successful 
proposal of \ac{QC-LDPC} codes or \ac{QC-MDPC} 
codes~\cite{baldi, MTSB13} as code families to build a secure and efficient 
instance of either the McEliece or the Niederreiter cryptosystem.

However, the efficient iterative algorithms used for decoding \ac{LDPC} and 
\ac{MDPC} codes are not bounded distance decoders, yielding a 
non-zero probability of obtaining a decoding failure, known as \ac{DFR}, 
which translates into a decryption failure rate for the corresponding code-based 
cryptosystems.
The presence of a non-null \ac{DFR} was shown to be exploitable by an active
adversary, which has access to a decryption oracle (the typical scenario
of a Chosen Ciphertext Attack (CCA)), to extract information on the secret
\ac{QC-LDPC} or \ac{QC-MDPC} code~\cite{Guo2016,Fabsic2017}.
To reliably avoid such attacks, cryptosystem constructions providing \ac{IND-CCA2}
guarantees, even when considering decoding failures, were analyzed in~\cite{HHK,Bindel2019}.
In order for the \ac{IND-CCA2} security guarantees to hold, the constructions
require that the average of the \ac{DFR} over all the keypairs, which an 
adversary is able to induce crafting messages, is below a given threshold $\delta$; 
a definition known as \emph{$\delta$-correctness}~\cite{HHK}.
Such a threshold $\delta$ must be exponentially small in the security parameter of the
scheme, in turn calling for requirements on the \ac{DFR} of the underlying 
code that cannot be estimated via numerical simulations (e.g., \ac{DFR}$\leq 2^{-128}$).

The impossibility of validating the \ac{DFR} through numerical simulations has
spurred significant efforts in modelling the behaviour of iterative decoders for 
\ac{QC-LDPC} and \ac{QC-MDPC} codes, with the goal of finding reliable tools to 
assess the DFR~\cite{Tillich2018, Baldi2019_CBC, Santini2019_ICC, Sendrier2019, 
Santini2019_Arxiv}.
A subset of the aforementioned works consider a very small number of iterations of the
decoder, providing code-specific exact bounds for the \ac{DFR}~\cite{Tillich2018, 
Santini2019_ICC, Santini2019_Arxiv}; however, employing such bounds to perform 
the code-parameter design results in impractically large public-key sizes.
In~\cite{Sendrier2019}, the authors adopt a completely different approach which 
extrapolates the \ac{DFR} in the desired regime from numerical simulations performed with higher \ac{DFR} values. This method assumes that the exponentially decreasing trend of the \ac{DFR} holds as the code length is increased while keeping the rate constant.
Such an assumption, however, does not rest on a theoretical basis.
Finally, in~\cite{Baldi2019_CBC} the authors characterize the \ac{DFR} of a 
two-iteration out-of-place decoder, providing a closed-form method to derive
an estimate of the average \ac{DFR} over all the \ac{QC-LDPC} codes with the same
length, rate and density, under the assumption that the bit-flipping decisions
taken during the first iteration are independent from each other.
In the recent work \cite{Drucker2019}, authors have highlighted an issue concerning possible \emph{weak keys} of \ac{QC-LDPC} and \ac{QC-MDPC} code-based cryptosystems, 
i.e., keypairs obtained from codes having a \ac{DFR} significantly lower than the average one.

\smallskip 
\noindent \textbf{Contributions.} We provide an analysis of the \ac{DFR}
of an in-place iterative \ac{BF}-decoder for \ac{QC-LDPC} and \ac{QC-MDPC}
codes acting on the estimated error locations in a randomized fashion 
for a fixed number of iterations.
We provide a closed form statistical model for such a decoder, allowing us
to derive a worst-case behaviour at each iteration, under clearly stated assumptions.
We provide both an analysis of the \ac{DFR} of the said decoder in the worst case scenario for the average \ac{QC-LDPC}/\ac{QC-MDPC} code, and we exploit the approach 
of~\cite{Santini2019_Arxiv} to derive a hard bound on the performance of the decoder 
on a given \ac{QC-LDPC}/\ac{QC-MDPC} code.
While our analysis on the behavior of a \ac{QC-LDPC}/\ac{QC-MDPC} code allows us
to match the requirements for a $\delta$-correct cryptosystem~\cite{HHK}, 
the hard bound we provide for the behavior of the decoder on a specific code 
allows us to discard \emph{weak keys} during the key generation phase, solving any concern about the use of weak keys.
We provide a confirmation of the effectiveness of our analysis by comparing its results with numerical simulations of the described in-place decoder.
 \section{Preliminaries}
\label{sec:preliminaries}
Throughout the paper, we will use uppercase (resp. lowercase) bold letters to 
denote matrices (resp. vectors).
Given a matrix $\bA$, its $i$-th row and $j$-th column are denoted as 
$\bA_{i,:}$  and $\bA_{:,j}$, respectively, while the entry on the 
$i$-th row, $j$-th column is denoted as $a_{i,j}$. 
Given a vector $\ba$, its length is denoted as $|\ba|$, while the $i$-th element 
is denoted as $a_i$, with $0\leq i \leq |a|-1$; finally, 
the support (i.e., the set of positions of the asserted elements in a sequence) 
and the Hamming weight of $\ba$ are denoted as $\supp{\ba}$ and 
$\weight{\ba}$, respectively.
We will use $\mathcal{P}_n$, $n \geq 1$, to denote the set of $n!$ permutations 
of $n$ elements, represented as a set of integers from $0$ to $n$$-$$1$, 
while the notation $\pi\xleftarrow{\$} \mathcal{P}_n$ is employed to randomly and 
uniformly pick an element in $\mathcal{P}_n$, denoting the picked permutation 
of integers in $\{0\ldots,n-1\}$ as $\pi$. 

As far as the cryptoschemes are concerned, in the following we will make use of 
a \ac{QC-LDPC}/\ac{QC-MDPC} code $\mathcal{C}$, with length 
$n=n_0p$, dimension $k=(n_0-1)p$ and redundancy $r=n-k=p$.
The private-key will coincide with the parity-check matrix
$\bH = \left[\bH_0,\bH_1,\cdots,\bH_{n_0-1}\right] \in \mathbb{F}_2^{r \times n}$,
where each $\bH_i$, $0\leq i\leq n_0-1$ is a binary circulant matrix of size 
$p\times p$ and fixed Hamming weight $v$ of each column/row.
Therefore, $\bH$ has constant column-weight $v$ and constant row-weight $w=n_0v$, 
and we say that $\bH$ is $(v,w)$-\textit{regular}.

When considering the case of the McEliece construction, the public-key may be chosen as the systematic generator matrix of the code.
The plaintext is in the form $\bc = \bm\bG+\be,\ \bc \in \mathbb{F}_2^{1\times n}$, where $\bm \in \mathbb{F}_2^{1\times n}$, $\be \in \mathbb{F}_2^{1\times n}$ and $\weight{\be}=t$. 
The decryption algorithm takes as input the ciphertext $\bc$ to compute the syndrome 
$\bs = \bc\bH^\top=\be\bH^\top, \ \bs \in \mathbb{F}_2^{1\times r}$, and the private-key $\bH$ to fed a syndrome decoding algorithm with both $\bs$ and $\bH$ and derive $\be$, from which the original message is recovered looking at the first $k$ elements of $\bc-\be$.

When the Niederreiter construction is considered, the public-key is defined as 
the systematic parity-check matrix of the code, obtained from the private-key as  
$\mathbf{M} = \mathbf{H}_0^{-1}\bH \in \mathbb{F}_2^{r \times n}$.
In this case, the message to be encrypted coincides with the error vector $\be \in \mathbb{F}_2^{1\times n}$, $\weight{\be}=t$, while 
the encryption algorithm computes the ciphertext  $\bc = \be\mathbf{M}^\top, \ \bc \in \mathbb{F}_2^{1\times r}$ as a syndrome.
The decryption algorithm takes as input the ciphertext $\bc$ and the private-key $\bH$ to compute a private-syndrome 
$\bs = \bc\bH_0^\top = \be\mathbf{M}^\top \bH_0^\top = \be\mathbf{H}^\top(\bH_0^\top)^{-1} \bH_0^\top = \be\bH^\top$ and, subsequently,  
fed with it a syndrome decoding algorithm to derive the original message $\be$.

 \section{Randomized In-place Bit-flipping Decoder}
\label{sec:rip}
In this section we describe a slightly modified version of the \ac{BF} decoder originally proposed by Gallager in 1963~\cite{Gal63}.
We focus on the \emph{in-place} \ac{BF}-decoder in which, at each bit evaluation, the decoder computes the number 
of unsatisfied parity-check equations in which the bit participates: when this number exceeds some threshold 
(which may be chosen according to different rules), then the bit is flipped and the syndrome is updated. 
Decoding proceeds until a null syndrome is obtained or a prefixed maximum number of iterations is reached. 

The algorithm we analyze is reported in Algorithm \ref{alg:RIP_BF}.
Inputs of the decoder are the binary parity-check matrix $\bH$, the syndrome $\bs$, the maximum 
number of iterations $\mathtt{itermax}$ and a vector $\bb$ of length $\mathtt{itermax}$, such that the $i$-th iteration uses $b_i$ as threshold.
The only difference with the classic in-place BF decoder is that the estimates on the error vector bits are processed in a random order, driven by a random permutation (generated at line $3$).
For this reason, we call this decoder \ac{RIP-BF} decoder.
Such a randomization, which is common to prevent side-channel 
analysis~\cite{DBLP:journals/ipl/AgostaBPS15,DBLP:books/daglib/0017272} (and 
typically goes by the name of instruction shuffling in that context), is 
crucial in our analysis, since it allows us to derive a worst case analysis, as 
we describe in the following section.

{\small
\begin{algorithm}[!t]
\LinesNumbered
\DontPrintSemicolon
\caption{Randomized In-Place BF decoder~\label{alg:RIP_BF}}
\KwIn{$\bs\in \Ftwo^{r}$: syndrome \newline
      $\bH\in \Ftwo^{r \times n}$: private parity-check matrix}
\KwOut{$\hat{\be}\in\Ftwo^n$: recovered error value\newline
       $\bs\in\Ftwo^r$: syndrome, null if error $\hat{\be}=\be$ }
\KwData{$\mathtt{itermax} \geq 1$: maximum number of (outer loop) iterations\newline
       $\bb=[b_1, \ldots, b_{\mathtt{itermax}}], b_k \in \{\lceil \frac{v}{2}\rceil, \ldots, v\},\, 1\leq k \leq \mathtt{itermax}$: flip thresholds}
\BlankLine        
$\mathtt{iter}\gets 0,\ \hat{\be}\gets \mathbf{0}_n$\;

\While{$(\mathtt{iter}<\mathtt{itermax})\wedge(\weight{\bs}>0)$}{
    \BlankLine 
    $\pi\xleftarrow{\$} \mathcal{P}_n$ \tcp*{random permutation of size $n$}
    \ForEach{$\hat{e}_j \in \pi(\hat{\be})$}{$\mathtt{upc}\gets 0$\;
        \For{$i\gets 0\ \mathbf{to}\ r-1$}{$\mathtt{upc}\gets\mathtt{upc}+(s_i \cdot h_{i,j})$\;
        }
        \If{$\mathtt{upc}\geq b_{\mathtt{iter}}$}{$\hat{e}_j\gets \hat{e}_j\oplus 1$ \tcp*{estimated error vector update}
            \For{$i\gets 0\ \mathbf{to}\ r-1$}{$s_i \gets s_i \oplus h_{i,j}$\;
            }
        }
    }
    $\mathtt{iter}\gets \mathtt{iter}+1$  \tcp*{update of the iterations counter}
}
\BlankLine 
\KwRet $\{\bs, \hat{\be}\}$ \;
\end{algorithm}
}

\subsection{Assessing Bit-flipping Probabilities}\label{sec:statistical_model}
In this section we describe a statistical approach to model the behaviour of the \ac{RIP-BF} decoder.
We assume that the bit evaluations are independent and uncorrelated, and depend 
only on the number of the bits of $\hat{\be}$ which do not match the ones of $\be$ 
at the beginning of the outer loop iterations.
Such an assumption is captured by the following statement.

Consider the execution of steps in Algorithm~\ref{alg:RIP_BF} from the beginning of an outer loop iteration (line 3). 
For each position $j$, with $0\leq j\leq n-1$, of the unknown error vector, \be, (or equivalently, for each column of the matrix \bH) if the number of the unsatisfied parity-checks (upc) influenced by $e_j$ exceeds the predefined threshold chosen for the current (outer loop) iteration, $\mathtt{b}_{\mathtt{iter}}$, then the $j$-th position  of the estimated error vector, $\hat{e}_j$, is flipped and the value of the syndrome is updated (lines 6-9).
Denoting as
\begin{itemize}
\item $\mathrm{P}_{f\mid1}=\mathtt{Prob}\left((j\text{-{\em th}}\ \mathtt{upc}) \geq b_{\mathtt{iter}}\mid e_j = 1 \right)$, the probability that 
the computation of the $j$-th upc yields an outcome greater or equal to the current threshold (thus, triggering a flip of $\hat{e}_j$) conditioned by the hyphotetical event of knowing that the actual $j$-th error bit is asserted, i.e., $e_j = 1$;  
\item $\mathrm{P}_{m\mid0}=\mathtt{Prob}\left((j\text{-{\em th}}\ \mathtt{upc}) < b_{\mathtt{iter}} \mid e_j = 0 \right)$, the probability that 
the computation of the $j$-th upc yields an outcome less than the current threshold (thus, maintaining the bit $\hat{e}_j$ unchanged) conditioned by the hyphotetical event of knowing that the actual $j$-th error bit is null, i.e., $e_j = 0$.
\end{itemize}
\noindent In the following analyses, the statement below is assumed to hold. 
\begin{ass}\label{ass:Pf_Pu} 
Both $\mathrm{P}_{f\mid1}$ and $\mathrm{P}_{m\mid0}$ are not a function of the bit-position in the actual error vector (i.e., $j$, in the previous formulae), although both probabilities are a function of the total number $\hat{t}=\weight{\be\oplus\hat{\be}}$ of positions over which the unknown error \be\ and the estimated error vector $\hat{\be}$ differ, at the beginning of the $j$-th inner loop iteration (line 5 in Algorithm~\ref{alg:RIP_BF}). 
\end{ass}
To derive closed formulae for both $\mathrm{P}_{f\mid1}$ and $\mathrm{P}_{m\mid0}$, we focus on \ac{QC-LDPC}/\ac{QC-MDPC} parity-check matrices as described in Section~\ref{sec:preliminaries} with column weight $v$ and row weight $w=n_0v$ and observe that Algorithm~\ref{alg:RIP_BF} uses the columns of the parity-check matrix, for each outer loop iteration, in an order that is chosen with a uniformly random draw (line 3), while the computation performed at lines 6--7 is independent by the processing order of each cell of the selected column. 
According to this, in the following we ``idelize'' the structure of the parity check-matrix, assuming each row of \bH\ independent from the others and modeled as a sample of a uniform random variable, distributed over all possible sequences of $n$ bits with weight $w$. More formally, 
\begin{ass}\label{ass:row_probability}
Let $\bH$ be a $r \times n$ quasi-cyclic block-circulant $(v,w)$-regular parity-check matrix and let $\bs$ be the $1 \times r$ syndrome corresponding to a $1\times n$ error vector $\be$ that is modeled as a sample from a uniform random variable distributed over the elements in $\mathbb{F}_2^{1\times n}$ with weight $t$.\\
We  assume that each row $\bh_{i, :}$, $0\leq i\leq r-1$, of the parity-check matrix \bH\ is well modeled as a sample from a uniform random variable distributed over the elements of $\mathbb{F}_2^{1\times n}$ with weight $w$.
\end{ass}
\begin{lemma}\label{lem:pr_flip_unflip} 
From Assumption~\ref{ass:row_probability}, the probabilities that the $i$-th bit of the syndrome $(0\leq i\leq r-1)$ is asserted knowing that the $z$-th bit of the error vector $(0\leq z\leq n-1)$ is null or not, i.e., $\pr{s_i = 1 | e_z} = \pr{\langle \bh_{i,:},\be \rangle = 1 | e_z}$, $\langle \bh_{i,:},\be \rangle  = \bigoplus_{j=0}^{n-1} h_{i,j} \cdot e_j$, can be expressed for each bit position $z$, $0\leq z \leq n-1$, of the error vector as follows:
$$
\pErrZeroUnsat=\pr{\langle \bh_{i,:},\be \rangle = 1\ |\ e_z=0} = \frac{\sum_{l=0,\text{ l odd}}^{\min\{w,t\}}\binom{w}{l}\binom{n-w}{t-l}}{\binom{n-1}{t}}
$$
$$
\pErrOneUnsat=\pr{\langle \bh_{i,:},\be \rangle = 1\ |\ e_z=1} = \frac{\sum_{l=0,\text{ l even}}^{\min\{w-1,t-1\}}\binom{w-1}{l}\binom{n-w}{t-1-l}}{\binom{n-1}{t-1}}
$$
Consequentially, the probability that Algorithm~\ref{alg:RIP_BF} performs a bit-flip of an element of the estimated error vector, $\hat{e}_z$, when the corresponding bit of the actual error vector is asserted, $e_z=1$, i.e., $\mathrm{P}_{f\mid1}$, and the  probability that Algorithm~\ref{alg:RIP_BF} maintains the value of the estimated error vector, $\hat{e}_z$, when the corresponding bit of the actual error vector is null, $e_z=0$, i.e., $\mathrm{P}_{m\mid0}$, are:
$$\mathrm{P}_{f\mid1}= \sum_{\upc = b}^{v}\binom{v}{\upc}\pic^{\upc}(1-\pic)^{v-\upc},$$
$$\mathrm{P}_{m\mid0} = \sum_{\upc = 0}^{b-1}\binom{v}{\upc}\pci^{\upc}(1-\pci)^{v-\upc}.$$
\end{lemma}
\begin{proof} Provided in Appendix~\ref{sec:appendix_d}.
\end{proof}

\subsection{Bounding Bit-flipping Probabilities for a Given Code}

Given a QC-LDPC code $\mathcal{C}$ with its $r\times n$ $(v,w)$-regular 
parity-check matrix $\bH$, let us consider each column of $\bH$, $\bh_{:, z}$, $0\leq z \leq n-1$, as a Boolean vector equipped with element-wise addition and multiplication denoted as  $\oplus$ and $\wedge$, respectively.  
Let $\mathbf{\Gamma}$ be 
the $n\times n$ integer matrix, where each element $\gamma_{x,y}\in \{0,\ldots,v\}$, with $0 \leq x,y \leq n-1$, is computed as the weight of the element-wise multiplication between two different columns, and 0 otherwise, i.e., 
$$\gamma_{x,y} = 
\begin{cases}
\weight{\bh_{:,x} \wedge \bh_{:,y}} & x \neq y \\
                                  0 & x = y    \\
\end{cases} 
$$
The integer matrix $\mathbf{\Gamma}$ is symmetric and, when derived from a block-circulant matrix, is made of circulant blocks, as well.

An alternate way of exhibiting the probability $\mathrm{P}_{f\mid1}$ that Algorithm~\ref{alg:RIP_BF} performs a bit-flip of an element of the estimated error vector, $\hat{e}_z$, when the corresponding bit of the actual error vector is asserted, i.e., $e_z=1$, consists in counting how many of the $\binom{n-1}{t-1}$ error vectors $\be$, with $e_z=1$, are such that the $z$-th upc counter computed employing the corresponding syndrome (see lines 6--7) is above the pre-defined threshold $b$:
\begin{equation}\label{eq:errorsets}
\mathrm{P}_{f\mid1} = \frac{|\,\{\be\ \text{s.t.}\ (z\text{-th}\ \mathtt{upc}) \geq b\}\,|}{\binom{n-1}{t-1}}.  
\end{equation}
Noting that the computation of $z$-th $\mathtt{upc}$ can be derived as a function of the unknown error vector $\be$ as follows:
$$ z\text{-th}\ \mathtt{upc} = v - \weight{ \bigoplus_{j\in\{\supp{\be}\setminus \{z\}\}} (\bh_{:,z} \wedge \bh_{:,j})}
 \geq v - \sum_{j\in\{\supp{\be}\setminus \{z\}\}} \gamma_{z,j}, 
$$
\noindent the following inequality concerning the numerator of the fraction in Eq.~\eqref{eq:errorsets} holds:
$$
|\,\left\{\be\ \text{s.t.}\ (z\text{-th}\ \mathtt{upc}) \geq b\right\}\,|
\ \geq \
\left|\,\left\{\be\ \text{s.t.}\ \left(v - \sum_{j\in\{\supp{\be}\setminus \{z\}\}}\gamma_{z,j}\right) \geq b \right\}\,\right|
$$
The cardinality of the set on the right-hand side of the above inequality asks for the counting 
of all error vectors such that the sum of the elements on the $z$-th row of the matrix $\mathbf{\Gamma}$ indexed by the positions in $\{\supp{\be}\setminus \{z\}\}$ (with $|\{\supp{\be}\setminus \{z\}\}| = t-1$) is less than $v-b$: i.e., $\sum_{j\in\{\supp{\be}\setminus \{z\}\}}\gamma_{z,j} \leq v-b$.
The answer to the said question is equivalent to counting the number of solutions of the corresponding {\em subset sum} problem~\cite{Cormen2009}, that is finding a subset of  $|\{\supp{\be}\setminus \{z\}\}| = t-1$ elements out of the ones in the row $\gamma_{z,:}$ adding up to at most $v-b$. A straightforward computation of such a counting is unfeasible for cryptographic relevant values of the involved parameters, exhibiting an exponential complexity in the correction capacity of the code $t$. 

However, observing that, for QC-LDPC codes, the number $\eta_z$ of \underline{unique} values on each row $\gamma_{z,:}$ of the matrix $\mathbf{\Gamma}$ is far lower than $t$, therefore we designed an algorithm computing the same result with a complexity exponential in $\eta_z$, reported in Appendix~\ref{sec:appendix_b}.
In the following, for the sake of conciseness, the outcome of the said algorithm fed with a row of the matrix $\mathbf{\Gamma}$, the cardinality $|\{\supp{\be}\setminus \{z\}\}| = t-1$ (i.e., the number of terms of the summation), and the $\mathtt{threshold}$ value that the sum must honor is denoted as: $\mathcal{N}(\gamma_{z,:}, t-1, \mathtt{threshold})$.
\begin{equation}\label{eq:lowerbound1}
\mathrm{P}_{f\mid1}\ \geq\  \frac{\displaystyle \max_{0\leq z \leq n-1}\left\{\mathcal{N}(\gamma_{z,:}, t-1, v-b)\right\}}{\binom{n-1}{t-1}}.
\end{equation}
With similar arguments, a lower bound on $\mathrm{P}_{m\mid0}$ can be derived, obtaining: 
\begin{equation}\label{eq:lowerbound2}
\mathrm{P}_{m\mid0}\ \geq\ \frac{\displaystyle \max_{0\leq z \leq n-1}\left\{\mathcal{N}(\gamma_{z,:}, t, b-1)\right\}}{\binom{n-1}{t}}.
\end{equation}
 \section{Modeling the DFR of the RIP-decoder}
\label{sec:evaluations}
Using the probabilities, $\mathrm{P}_{f\mid1},\mathrm{P}_{m\mid0}$, that we have
derived in the previous section, under Assumption~\ref{ass:Pf_Pu} we can 
derive a statistical model for the RIP-BF decoder.
To this end, we now focus on a single iteration of the outer loop of 
Algorithm~\ref{alg:RIP_BF}.
In particular, as we describe next, we consider a \emph{worst-case} evolution 
for the decoder, by assuming that, at each iteration of the inner loop, 
it evolves through a path that ends in the a decoding success with the lowest 
probability.
We obtain a decoding success if the decoder terminates the inner loop iteration 
in the state where the estimate of the error $\hat{\be}$ matches the actual 
error $\be$. Indeed, in such a case, we have $\weight{\be\oplus \hat{\be}}=0$.

Let $\bar{\be}$ be the error estimate at the beginning of the outer loop of 
Algorithm~\ref{alg:RIP_BF} (line 3), and $\hat{\be}$ be the error estimate at the 
beginning of the inner loop of the same algorithm (line 5).
In other words, $\bar{\be}$ is a snapshot of the error estimate made by the RIP
decoder before a sweep of $n$ estimated error bit evaluations is made, while
$\hat{\be}$ is the value of the estimated error vector before each estimated 
error bit is evaluated.

Let $\hat{t}$ denote the number of residual erroneous bit estimations at the 
beginning of the inner loop iteration, that is $\hat{t}=\weight{\be \oplus \hat{\be}}$.
From now on, we highlight the dependency of $\mathrm{P}_{f\mid1}$ and 
$\mathrm{P}_{m\mid0}$ from the current value of $\hat{t}$, writing them down as
$\Pfone{\hat{t}}$ and $\Pmzero{\hat{t}}$.

We denote as $\pi$ the permutation picked in line $3$ of Algorithm~\ref{alg:RIP_BF}.
Let $\mathcal{P}^*_n$ be the set of all permutations $\pi^* \in \mathcal{P}^*_n$ such that
$$\supp{ \pi^*(\be)\oplus \pi^*(\bar{\be}) }=\{n-\hat{t},n-\hat{t}+1,\cdots,n-1\},\hspace{2mm}\forall \pi^*\in\mathcal{P}_n^*.$$

Let $\pr{\left.\hat{\be} \neq \be\right| \hspace{1mm}\pi\in\mathcal{P}_n}$ be 
the probability that the estimated error vector $\hat{\be}$ at the end of the 
current inner loop iteration is different from $\be$, conditioned by the fact 
that the permutation $\pi$ was applied at the beginning of the outer loop.
Similarly, we define 
 $\pr{\left.\hat{\be} \neq \be\right| \hspace{1mm}\pi^*\in\mathcal{P}^*_n}$.
Note that it can be verified that $\Pf{\hat{t}}\geq \Pf{\hat{t}+1}, 
\Pu{\hat{t}}\geq \Pu{\hat{t}+1},\hspace{1mm}\forall \hat{t}$, as increasing the 
number of current mis-estimated error bits, increases the likelihood of a wrong 
decoder decision.
By leveraging the assumption made in the previous section, we now prove that 
the decoder reaches a correct decoding at the end of the outer loop with the 
least probability each time a $\pi^* \in \mathcal{P}^*_n$ is applied at the
beginning of the outer loop.
\begin{lemma}\label{lem:worst_case} The execution path of the inner loop in 
Algorithm~\ref{alg:RIP_BF} yielding the worst possible decoder success rate is
the one taking place when $\pi^* \in \mathcal{P}^*_n$ is applied at the
beginning of the outer loop, that is:
$$\forall \pi\in\mathcal{P}_n, \forall \pi^*\in\mathcal{P}_n^*, \ \ \pr{\left.\hat{\be} \neq \be\right| \hspace{1mm}\pi\in\mathcal{P}_n}
\leq \pr{\left.\hat{\be} \neq \be\right| \hspace{1mm}\pi^*\in\mathcal{P}^*_n}.$$
\end{lemma}
\begin{proof}
See Appendix~\ref{sec:appendix_c}.
\end{proof}

From now on we will assume that, in each iteration, a permutation from the set 
$\mathcal{P}_n^*$ is picked; in other words, we are assuming that the decoder is
always constrained to reach a decoding success through the worst possible 
execution path.
Let us define the following two sets: $E_1 = S(\be)$, and $E_0 = \{0,\ldots,n-1\} \setminus S(\be)$.
Denote with $\hat{t}_0 = \left|\left\{
S(\be\oplus \bar{\be}) \cap E_0 \right\} \right|$, that is the number of places where the estimated error 
at the beginning of the outer loop iteration $\bar{\be}$ differs from the actual 
$\be$, in positions included in $E_0$. Analogously, define $\hat{t}_1 = \left|\left\{
S(\be\oplus \bar{\be}) \cap E_1 \right\} \right|$.
Furthermore, let
\begin{enumerate}[i)]
\item $\prJ{\omega\xrightarrow{E_0}x}$ denote the probability that the decoder 
in Algorithm \ref{alg:RIP_BF}, starting from a state where $\weight{\hat{\be}\oplus\be} = \omega$, 
and acting in the order specified by a worst case permutation $\pi^*\in\mathcal{P}_n^*$ 
ends in a state with $\hat{t}_0 = x$ after completing the inner loop at lines $4$ -- $11$;
\item $\prJ{\omega\xrightarrow{E_1}x}$ denote the probability that the decoder 
in Algorithm \ref{alg:RIP_BF}, starting from a state where $\weight{\hat{\be}\oplus\be} = \omega$, 
and acting in the order specified by a worst case permutation $\pi^*\in\mathcal{P}_n^*$ 
ends in a state with $\hat{t}_1 = x$ residual errors among the 
bits indexed by $E_1$ after completing the loop at lines $4$--$11$;

\item $\prJ{\omega\xrightarrow[i]{} x}$  as the probability that, 
starting from a state such that
$\weight{\hat{\be}\oplus\be}=\omega$, after $i$ iterations 
the outer loop at lines $2$--$12$ of Algorithm \ref{alg:RIP_BF}, each one 
operating with a worst case permutation, ends in a state where
$\weight{\hat{\be}\oplus\be} = x$.
\end{enumerate}
The expressions of the probabilities \emph{i)} and \emph{ii)} are derived in Appendix~\ref{sec:appendix_a}, 
and only depend on the probabilities $\Pfone{\hat{t}}$ and $\Pmzero{\hat{t}}$. 

We now describe how the aforementioned probabilities can be used to express the worst case \ac{DFR} after $\mathtt{itermax}$ iterations, which we denote as $\mathrm{DFR}^*_{\mathtt{itermax}}$.
First of all, we straightforwardly have
\begin{align}
\prJ{\omega\xrightarrow[1]{}x} & = \sum_{\delta = \max\{0 \hspace{1mm};\hspace{1mm} x-(n-\omega) \}}^t  \prJ{\omega \xrightarrow{E_0}x-\delta}\prJ{\omega\xrightarrow{E_1}\delta}.
\end{align}
We can denote as $\hat{t}^{(i)}=\weight{\be\oplus \hat{\be}^{(\text{\tt iter})}}$, that is: $\hat{t}^{(i)}$ corresponds to the number of residual errors after the $i$-th outer loop iteration.
Then, by considering all possible configurations of such values, and taking into account that the first iteration begins with $t$ residual errors, we have
{\small
\begin{align}
\prJ{t\xrightarrow[\mathtt{itermax}-1]{} \hat{t}^{(\mathtt{itermax}-1)}} =  \sum_{\tleft{0}=0}^{n}\cdots \nonumber & \sum_{\tleft{\mathtt{itermax}-2}=0}^{n}
\prJ{\tleft{\mathtt{i}_{\mathtt{max}-2}}\xrightarrow[1]{} \hat{t}^{(\mathtt{itermax}-1)}}\\
& \prod_{j=0}^{\mathtt{itermax}-2}\prJ{\hat{t}^{(j-1)}\xrightarrow[1]{}\hat{t}^{(j)}},
\end{align}
}
where, to have a consistent notation, we consider $\tleft{-1} = t$.
The above formula is very simple and, essentially, takes into account all possible transitions starting from an initial number of residual errors equal to $t$ and ending in $x$ residual errors.
Taking this probability into account, the \ac{DFR} after $\mathtt{itermax}$ iterations is straightforwardly obtained as
\begin{equation}
\label{eq:dfr_imax}
\mathrm{DFR}^*_{\mathtt{itermax}} = 1 - \sum_{\hat{t}^{(\mathtt{itermax}-1)}=0}^n \prJ{t\xrightarrow[\mathtt{itermax}-1]{} \hat{t}^{(\mathtt{itermax}-1)}}\prJ{\hat{t}^{(\mathtt{itermax}-1)}\xrightarrow[1]{}0}.
\end{equation}

\subsection{Analyzing a Single-iteration Decoder}
For the case of the decoder performing just one iteration, the simple expression 
of the DFR has been derived in the proof of Lemma \ref{lem:worst_case}, that is
\begin{equation}
\label{eq:dfr_1}
\mathrm{DFR}_1^* \nonumber  = 1 - \prJ{t\xrightarrow[1]{} 0}  = \bigg(\Pu{t}\bigg)^{n-t}\prod_{j=1}^t\Pf{j}.
\end{equation}

 Actually, for just one iteration, the average \ac{DFR} (corresponding to the use 
of a random permutation $\pi$) can be approximated in a very simple way, 
as follows.
Let $a_i, a_{i+1}$, with $i\in[0 ; t-2]$, be two consecutive elements of 
$\supp{\pi(\be)}$.
Then denote with $d$ the average zero-run lenght in $\be$, $d= \mathbb{E} 
\left[ a_{i+1}-a_i\right] = \frac{n-t}{t+1},\hspace{2mm}\forall  i \in [0 ; t-2]$
where $\mathbb{E}[\cdot]$ denotes the expected value.
Consequently, we can write
\begin{equation}
\mathrm{DFR}_1 \approx 1- \left(\prod_{j=1}^{t}\bigg(\Pu{j}\bigg)^d \right) \prod_{l=1}^t\Pf{l}.
\end{equation}

\subsection{Simulation Results}
\begin{figure}[!t]
\centering
\subfigure[One Iteration]{\begin{tikzpicture}[scale=0.65]
  \begin{axis}[xtick={0,10,...,90},
               legend columns=2,
               xmin = 10,
               xmax= 90,
               grid = major,
               ymode=log,
               ymin=0.000001,
               ymax=1,
               legend style={at={(0.5,-0.4)},anchor=south},
               mark size=2pt,
               xlabel={$t$}, 
               ylabel={DFR}, 
               ]

\addplot+[blue, densely dotted, no marks,line width=1.2pt] plot coordinates{
(80, 1)
(79, 0.9999999999)
(78, 0.9999999992)
(77, 0.9999999952)
(76, 0.9999999743)
(75, 0.9999998755)
(74, 0.999999457)
(73, 0.9999978583)
(72, 0.9999923323)
(71, 0.9999749891)
(70, 0.9999253919)
(69, 0.9997956797)
(68, 0.9994842925)
(67, 0.998795572)
(66, 0.9973867767)
(65, 0.9947115353)
(64, 0.9899774422)
(63, 0.9821411366)
(62, 0.9699620425)
(61, 0.9521236799)
(60, 0.9274115362)
(59, 0.8949159061)
(58, 0.8542153872)
(57, 0.8054975839)
(56, 0.7495881369)
(55, 0.6878820247)
(54, 0.6221940683)
(53, 0.5545615609)
(52, 0.4870373884)
(51, 0.4215074612)
(50, 0.3595551968)
(49, 0.3023825963)
(48, 0.2507858305)
(47, 0.2051753186)
(46, 0.1656265298)
(45, 0.1319475132)
(44, 0.1037512575)
(43, 0.08052417826)
(42, 0.06168534976)
(41, 0.04663395979)
(40, 0.03478459315)
(39, 0.02559132618)
(38, 0.01856234989)
(37, 0.01326709828)
(36, 0.009337806644)
(35, 0.006467199185)
(34, 0.004403704998)
(33, 0.002945292183)
(32, 0.001932727112)
(31, 0.001242826756)
(30, 0.0007820801679)
(29, 0.000480867465)
(28, 0.0002883948399)
(27, 0.0001683849953)
(26, 0.9550749113e-4)
(25, 0.5249726322e-4)
(24, 0.2788761413e-4)
(23, 0.1427286173e-4)
(22, 0.7012852798e-5)
(21, 0.3294530539e-5)
(20, 0.1472868792e-5)
};
\addlegendentry{Est. $\mathrm{DFR}_1$, };

  \addplot[blue, only marks, mark=x,mark size=2.5pt]
             table [x=t,y expr=\thisrow{rip_1it_failures}/\thisrow{tot_tx},col sep = comma]{rip1_p4801_v45_b35.csv};
\addlegendentry{Sim. $\mathrm{DFR}_1$};

  \addplot[cyan, densely dotted,line width=1.2pt]
             table [x=t,y=rip_1it,col sep = comma]{f_synth_rip_p4801_v45_b35.csv};
  \addlegendentry{Est. $\mathrm{DFR}^*_1$};

\addplot[cyan, only marks, mark=x,mark size=2.5pt]
             table [x=t,y expr=\thisrow{rip_2it_wc_failures}/\thisrow{tot_tx},col sep = comma]{rip1_wc_p4801_v45_b35.csv};
\addlegendentry{Sim. $\mathrm{DFR}^*_1$};

\addplot[green!70!blue,densely dotted,line width=1.2pt]
           table [x=t,y=rip_1it,col sep = comma]{synth_ripes_p4801_v45_b35.csv};
\addlegendentry{Est. bound $\mathrm{DFR}^*_1\phantom{mm}$};

\addplot[black, only marks, mark=*,mark size=1.2pt]
             table [x=t,y expr=\thisrow{ip_1it_failures}/\thisrow{tot_tx},col sep = comma]{ip_p4801_v45_b35.csv};
\addlegendentry{Sim. $\mathrm{DFR}_{I}$};

\end{axis}
\end{tikzpicture}
 }
\subfigure[Two Iterations]{\begin{tikzpicture}[scale=0.65]
  \begin{axis}[xtick={10,20,...,90},
               xmax=90,
               xmin=10,
               legend columns=2,
               grid = major,
               legend style={at={(0.5,-0.4)},anchor=south},
               ymode=log, 
               ymin=0.000001,ymax=1,
               mark size=2pt,
               xlabel={$t$},
               ylabel={DFR}]

  \addplot[cyan,densely dotted, no marks,line width=1.2pt]
             table [x=t,y=rip_1it,col sep = comma]{synth_rip2_p4801_v45_b35.csv};
  \addlegendentry{Est. $\mathrm{DFR}^*_2$};
  
  \addplot[cyan, only marks, mark=x,mark size=2.5pt]
             table [x=t,y expr=\thisrow{rip_2it_wc_failures}/\thisrow{tot_tx},col sep = comma]{rip2_wc_p4801_v45_b35.csv};
\addlegendentry{Sim. $\mathrm{DFR}^*_2$};

  \addplot[green!70!blue,densely dotted, no marks,line width=1.2pt]
             table [x=t,y=rip_1it,col sep = comma]{synth_rip2es_p4801_v45_b35.csv};
  \addlegendentry{Est. bound $\mathrm{DFR}^*_2\phantom{mm}$};

\addplot[blue, only marks, mark=x,mark size=2.5pt]
             table [x=t,y expr=\thisrow{rip_2it_failures}/\thisrow{tot_tx},col sep = comma]{rip_p4801_v45_b35.csv};
\addlegendentry{Sim. $\mathrm{DFR}_2$};

\end{axis}
\end{tikzpicture}
 }
\caption{Experimental validation of the \ac{DFR} estimates (Est.) through numerical simulations (Sim.). The QC-LDPC code parameters are $n_0=2$, $p=4,801$ and $v=45$. The decoding threshold is $b_0=25$.
\label{fig:sim}}
\end{figure}
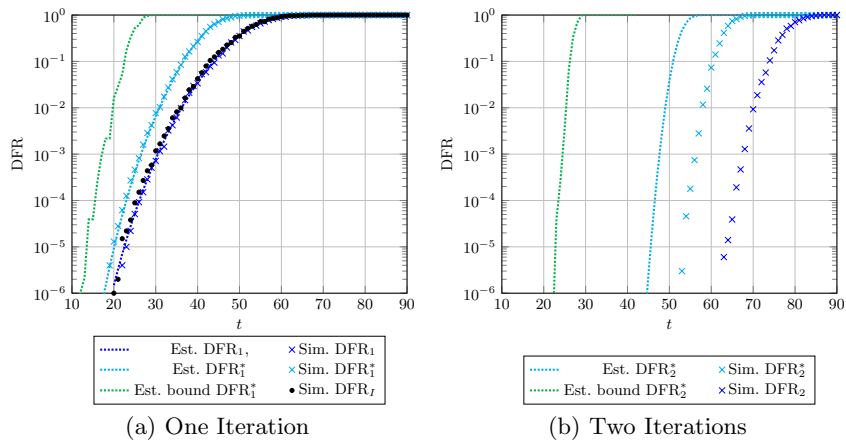
In this section we report the results of an experimental validation of the proposed analysis of the behavior of the RIP decoder.
As a case study we chose a QC-LDPC code having the parity
check matrix $\bH$ formed by $n_0=2$ circulant blocks of size $p=4801$, 
column weight $v=45$ and we assessed the \ac{DFR}
varying the error weight $t$ from $10$ to $100$, attempting to decode $10^6$
error vectors for each value of the error weight.
To this end, we implemented the RIP decoder in \texttt{C99}, and run the experiments
on an Intel Core i$5$-$6500$ CPU running at $3$.$20$ GHz, compiling the code with
the \texttt{GCC} $8$.$3$.$0$ and running the built executables on 
Debian GNU/Linux $10$.$2$ (stable).
The computation of the worst case DFR estimates and bounds in 
Eq.~\eqref{eq:lowerbound1} and Eq.~\eqref{eq:lowerbound2} were realized 
employing the NTL library~\cite{NTL}, while the solver for the counting version of the 
subset sum problem was implemented in plain \texttt{C++}.
Computing the entire DFR upper bound relying on the counting subset sum problem
takes significantly less than a second, for the selected parameters.
We report the results considering a bit flipping threshold of $b=25$, for all
the iterations; however, we obtained analogous results varying the bit flipping
threshold. The results with thresholds different from $25$ are omitted for lack
of space.
Figure~\ref{fig:sim} reports the results of numerical simulations of the \ac{DFR}
of the RIP decoder running for either one or two iterations, while employing a
random permutation ($\mathrm{DFR}_1$ and $\mathrm{DFR}_2$) or artificially 
computing the error estimates according to the worst-case permutation 
($\mathrm{DFR}^*_1$ and $\mathrm{DFR}^*_2$).
As it can be seen, our technique for the \ac{DFR} estimation provides a perfect 
match for the case of a single iteration, while our assumptions turn out to 
provide a conservative estimate for the worst-case \ac{DFR} in the case of a 
$2$-iteration RIP decoder.
In both cases, the actual behavior of the decoder with a random permutation
matches our expectations of having the \ac{DFR} bounded by both the worst-case 
one and the closed form code-specific bound reported in green in Fig.~\ref{fig:sim}.
Finally, it is interesting to note, from an implementation viewpoint, that
skipping the permutation in the case of a single-iteration RIP decoder appears to
have no effect on the simulated DFR (black dots, marked $\mathrm{DFR}_{I}$ in 
Fig.~\ref{fig:sim}). This can be explained observing that the first iteration of
the RIP decoder is actually applying the random permutation to the positions of 
the error estimates which have randomly-placed discrepancies with the actual error itself.
 \section{Conclusions}
We provided a statistical analysis of the behavior of a randomized
in place bit flipping decoder, derived from the canonical one by randomizing the
order in which the estimated error positions are processed.
This modification to the decoder allows us to provide a statistical worst-case
analysis of the DFR of the decoder at hand, both considering the average behavior among all the codes with the same length, dimension and number of parity checks, and a code-specific bound for a given QC-LDPC/QC-MDPC.
The former analysis can be fruitfully exploited to design code parameters 
allowing to obtain DFR values such as the ones needed to employ QC-LDPC/QC-MDPC
codes in constructions providing IND-CCA2 guarantees under the assumption 
that the underlying scheme is $\delta$-correct~\cite{HHK}.
The latter result allows us to analyze a given QC-LDPC/QC-MDPC code to assess whether the \ac{DFR} it exhibits is above the maximum tolerable one for an IND-CCA2 construction, thus allowing us to discard weak keypairs upon generation.
We note that our analysis relies on the RIP decoder performing a finite number of
iterations, as opposed to the one provided in~\cite{Sendrier2019}, in turn 
allowing a constant-time implementation of the RIP decoder itself.
This fact is of significant practical relevance since the timing information leaked from decoders performing a variable number of iterations was shown to be as valuable
as the one leaked by decryption failures to a CCA 
attacker~\cite{Eaton2018,Santini2019_react}, leading to concrete violations of the IND-CCA2 property.

 \bibliographystyle{abbrv}

\begin{thebibliography}{10}

\bibitem{DBLP:journals/ipl/AgostaBPS15}
G.~Agosta, A.~Barenghi, G.~Pelosi, and M.~Scandale.
\newblock {Trace-based schedulability analysis to enhance passive side-channel
  attack resilience of embedded software}.
\newblock {\em Inf. Process. Lett.}, 115(2):292--297, 2015.

\bibitem{Baldi2019_CBC}
M.~Baldi, A.~Barenghi, F.~Chiaraluce, G.~Pelosi, and P.~Santini.
\newblock {LEDAcrypt: {QC-LDPC} Code-Based Cryptosystems with Bounded
  Decryption Failure Rate}.
\newblock In M.~Baldi, E.~Persichetti, and P.~Santini, editors, {\em Code-Based
  Cryptography - 7th International Workshop, {CBC} 2019, Darmstadt, Germany,
  May 18-19, 2019, Revised Selected Papers}, volume 11666 of {\em Lecture Notes
  in Computer Science}, pages 11--43. Springer, 2019.

\bibitem{baldi}
M.~Baldi, F.~Chiaraluce, R.~Garello, and F.~Mininni.
\newblock {Quasi-Cyclic Low-Density Parity-Check Codes in the {McEliece}
  Cryptosystem}.
\newblock In {\em Proceedings International Conference on Communications (ICC
  2007)}, pages 951--956, Glasgow, Scotland, Jun. 2007.

\bibitem{DBLP:journals/tit/BerlekampMT78}
E.~R. Berlekamp, R.~J. McEliece, and H.~C.~A. van Tilborg.
\newblock {On the inherent intractability of certain coding problems
  (Corresp.)}.
\newblock {\em {IEEE} Trans. Information Theory}, 24(3):384--386, 1978.

\bibitem{Ber10}
D.~J. Bernstein.
\newblock Grover vs. {McEliece}.
\newblock In {\em Proceedings Post-Quantum Cryptography: Third International
  Workshop (PQCrypto 2010)}, pages 73--80, Darmstadt, Germany, May 2010.
  Springer Berlin Heidelberg.

\bibitem{Bindel2019}
N.~Bindel, M.~Hamburg, K.~Hövelmanns, A.~Hülsing, and E.~Persichetti.
\newblock Tighter proofs of {CCA} security in the quantum random oracle model.
\newblock Cryptology ePrint Archive, Report 2019/590, 2019.
\newblock \url{https://eprint.iacr.org/2019/590}.

\bibitem{Cormen2009}
T.~H. Cormen, C.~E. Leiserson, R.~L. Rivest, and C.~Stein.
\newblock {\em {Introduction to Algorithms, Third Edition}}.
\newblock The MIT Press, 3rd edition, 2009.

\bibitem{Drucker2019}
N.~Drucker and S.~Gueron.
\newblock A toolbox for software optimization of {QC-MDPC} code-based
  cryptosystems.
\newblock Cryptology ePrint Archive, Report 2017/1251, 2017.
\newblock \url{https://eprint.iacr.org/2017/1251}.

\bibitem{Eaton2018}
E.~Eaton, M.~Lequesne, A.~Parent, and N.~Sendrier.
\newblock {QC-MDPC}: A timing attack and a {CCA}2 {KEM}.
\newblock In T.~Lange and R.~Steinwandt, editors, {\em PQCrypto}, pages 47--76,
  Fort Lauderdale, FL, USA, Apr. 2018. Springer International Publishing.

\bibitem{Fabsic2017}
T.~Fab{\v{s}}i{\v{c}}, V.~Hromada, P.~Stankovski, P.~Zajac, Q.~Guo, and
  T.~Johansson.
\newblock {A Reaction Attack on the {QC-LDPC} {McEliece} Cryptosystem}.
\newblock In T.~Lange and T.~Takagi, editors, {\em Post-Quantum Cryptography:
  8th International Workshop, PQCrypto 2017}, pages 51--68. Springer, Utrecht,
  The Netherlands, June 2017.

\bibitem{faugere}
J.-C. Faug{\`e}re, A.~Otmani, L.~Perret, and J.-P. Tillich.
\newblock {Algebraic Cryptanalysis of {McEliece} Variants with Compact Keys}.
\newblock In H.~Gilbert, editor, {\em EUROCRYPT}, volume 6110 of {\em Lecture
  Notes in Computer Science}, pages 279--298. Springer, 2010.

\bibitem{Gal63}
R.~G. Gallager.
\newblock {\em Low-Density Parity-Check Codes}.
\newblock PhD thesis, M.I.T., 1963.

\bibitem{Guo2016}
Q.~Guo, T.~Johansson, and P.~Stankovski.
\newblock A key recovery attack on {MDPC} with {CCA} security using decoding
  errors.
\newblock In J.~H. Cheon and T.~Takagi, editors, {\em ASIACRYPT 2016}, volume
  10031 of {\em LNCS}, pages 789--815. Springer Berlin Heidelberg, 2016.

\bibitem{HHK}
D.~Hofheinz, K.~H{\"{o}}velmanns, and E.~Kiltz.
\newblock {A Modular Analysis of the Fujisaki-Okamoto Transformation}.
\newblock In Y.~Kalai and L.~Reyzin, editors, {\em Theory of Cryptography -
  15th International Conference, {TCC} 2017, Baltimore, MD, USA, November
  12-15, 2017, Proceedings, Part {I}}, volume 10677 of {\em Lecture Notes in
  Computer Science}, pages 341--371. Springer, 2017.

\bibitem{DBLP:books/daglib/0017272}
S.~Mangard, E.~Oswald, and T.~Popp.
\newblock {\em Power analysis attacks - revealing the secrets of smart cards}.
\newblock Springer, 2007.

\bibitem{McE78}
R.~J. {McEliece}.
\newblock A public-key cryptosystem based on algebraic coding theory.
\newblock {\em Deep Space Network Progress Report}, 44:114--116, Jan. 1978.

\bibitem{MTSB13}
R.~Misoczki, J.-P. Tillich, N.~Sendrier, and P.~L. Barreto.
\newblock {MDPC}-{McEliece}: New {McEliece} variants from moderate density
  parity-check codes.
\newblock In {\em Proceedings {IEEE} International Symposium on Information
  Theory (ISIT 2013)}, pages 2069--2073, Istambul, Turkey, Jul. 2013.

\bibitem{nied}
H.~Niederreiter.
\newblock Knapsack-type cryptosystems and algebraic coding theory.
\newblock {\em Problems of Control and Information Theory}, 15(2):159--166,
  1986.

\bibitem{SALOMAA196971}
A.~Salomaa.
\newblock {Chapter II - Finite Non-deterministic and Probabilistic Automata}.
\newblock In A.~Salomaa, editor, {\em Theory of Automata}, volume 100 of {\em
  International Series of Monographs on Pure and Applied Mathematics}, pages 71
  -- 113. Pergamon, 1969.

\bibitem{Santini2019_ICC}
P.~{Santini}, M.~{Battaglioni}, M.~{Baldi}, and F.~{Chiaraluce}.
\newblock {Hard-Decision Iterative Decoding of {LDPC} Codes with Bounded Error
  Rate}.
\newblock In {\em Proceedings IEEE International Conference on Communications
  (ICC 2019)}, pages 1--6, Shanghai, China, May 2019.

\bibitem{Santini2019_Arxiv}
P.~Santini, M.~Battaglioni, M.~Baldi, and F.~Chiaraluce.
\newblock A theoretical analysis of the error correction capability of {LDPC}
  and {MDPC} codes under parallel bit-flipping decoding, 2019.

\bibitem{Santini2019_react}
P.~Santini, M.~Battaglioni, F.~Chiaraluce, and M.~Baldi.
\newblock Analysis of reaction and timing attacks against cryptosystems based
  on sparse parity-check codes.
\newblock In {\em CBC}, 2019.

\bibitem{Sendrier2019}
N.~Sendrier and V.~Vasseur.
\newblock {On the Decoding Failure Rate of {QC-MDPC} Bit-Flipping Decoders}.
\newblock In J.~Ding and R.~Steinwandt, editors, {\em Post-Quantum Cryptography
  - 10th International Conference, PQCrypto 2019, Chongqing, China, May 8-10,
  2019 Revised Selected Papers}, volume 11505 of {\em Lecture Notes in Computer
  Science}, pages 404--416. Springer, 2019.

\bibitem{NTL}
V.~Shoup.
\newblock {NTL: A Library for doing Number Theory}.
\newblock \url{http://shoup.net/ntl/}, Version 11.4.1, 2019.

\bibitem{Tillich2018}
J.~Tillich.
\newblock {The Decoding Failure Probability of {MDPC} Codes}.
\newblock In {\em 2018 {IEEE} International Symposium on Information Theory,
  {ISIT} 2018, Vail, CO, USA, June 17-22, 2018}, pages 941--945. {IEEE}, 2018.

\end{thebibliography}

\appendix
\section{Deriving the Bit-flipping Probabilities for the RIP Decoder}
\label{sec:appendix_a}
Denote with $\hat{t}_0 = \left|\left\{S(\be\oplus \bar{\be}) \cap E_0 \right\} \right|$, 
that is the number of places where the estimated error at the beginning of the 
outer loop iteration $\bar{\be}$ differs from the actual $\be$, in positions 
included in $E_0$. Analogously, define $\hat{t}_1 = \left|\left\{S(\be\oplus \bar{\be}) \cap E_1 \right\} \right|$.

We now characterize the statistical distribution of $\hat{t}_0$ and $\hat{t}_1$ after
$n$ iterations of the inner loop of the RIP-BF decoder are run, processing the
estimated error bit positions in the order pointed out by $\pi^*\in\mathcal{P}_n^*$, i.e., the permutation which places at the end all the positions $j$ where $\hat{e}_j \neq e_j$.
We point out that, at the first iteration of the outer loop of the decoder, 
this coincides with placing all the positions where $e_j=1$ at the end, since 
$\bar{\be}$ is initialized to the $n$-binary elements zero vector, hence $\bar{\be}\oplus \be=\be$.

In characterizing the distribution of $\hat{t}_0$, because of Assumption \ref{ass:Pf_Pu}, we rely only on the probabilities $\Pfzero{t}$
and $\Pmzero{t}$, i.e. the probability that an error estimate bit will be flipped
or maintained.
In the following, for the sake of simplicity, we will consider $\hat{t}_1=t$, which
is the case of the RIP decoder performing the first outer loop iteration.
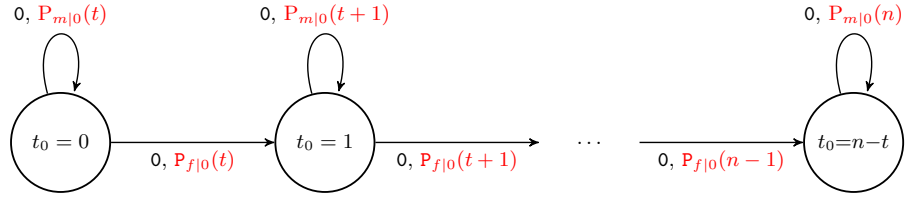
\begin{figure}[t]
    \centering
    \resizebox{\linewidth}{!}{
\begin{tikzpicture}[->, >=stealth', auto, semithick, node distance=4cm]

\tikzstyle{every state}=[fill=white,draw=black,thick,text=black,scale=1, inner sep=0pt, minimum size=1.5cm]
\tikzstyle{exit_state}=[fill=white,draw=myred,thick,text=myred,scale=1]
\node[state]    (A)                     {$t_0 = 0$};

\node[state]    (B)[right of=A]   {$t_0 = 1$};
\node[minimum size=1.5cm]         (C)[right of=B]   {$\ldots$};
\node[state]    (D)[right of=C]   {$t_0$$=$$n$$-$$t$};

\path (A)    edge  node[below]{\texttt{0}, \textcolor{red}{$\Pfzero{t}$}}     (B);
\path (B)    edge  node[below]{\texttt{0}, \textcolor{red}{$\Pfzero{t+1}$}}     (C);
\path (C)    edge  node[below]{\texttt{0}, \textcolor{red}{$\Pfzero{n-1}$}}     (D);

\path (A)    edge [loop above] node{\texttt{0}, \textcolor{red}{$\Pmzero{t}$}}     (A);
\path (B)    edge [loop above] node{\texttt{0}, \textcolor{red}{$\Pmzero{t+1}$}}     (B);
\path (D)    edge [loop above] node{\texttt{0}, \textcolor{red}{$\Pmzero{n}$}}     (D);
\end{tikzpicture}
 }
    \caption{Structure of the probabilistic FSA modeling the evolution of the 
    distribution of the $\hat{t}_0$ variable. Read characters are reported in black, 
    transition probabilities in red.}
    \label{fig:t0_evolution}
\end{figure}
To model the statistical distribution of $\hat{t}_0$ we employ the framework of Probabilistic
Finite State Automata (PFSA)~\cite{SALOMAA196971}.
Informally, a PFSA is a Finite State Automaton (FSA) characterized by transition 
probabilities for each of the transitions of the FSA.
The state of a PFSA is a discrete probability distribution over the set of 
FSA states and the probabilities of the transitions starting
from the same FSA state, reading a the same symbol, must add up to one.

We model the statistical distribution of $\hat{t}_0$ as the state of a PFSA having 
$n-t$ FSA states, each one mapped onto a specific value for $\hat{t}_0$, as depicted in 
Figure~\ref{fig:t0_evolution}.
We consider the underlying FSA to be accepting the input language constituted
by binary strings obtained as the sequences of $\hat{e}_j \neq e_j$ values, 
where $j$ is the error estimate position being processed by the RIP decoder
at a given inner loop iteration.
We therefore have that, for the PFSA modeling the evolution of $\hat{t}_0$ 
while the RIP decoder acts on the first $n-t$ positions specified by $\pi^*$, 
all the read bits will be equal to $0$, as $\pi^*$ sorts the positions of 
$\hat{\be}$ so that the ($n-t$ at the first iteration) positions with no 
discrepancy between $\bar{\be}$ and $\be$ come first.

The transition probability for the PFSA transition from a state 
$\hat{t}_0=i$ to $\hat{t}_0=i+1$ requires the RIP decoder to
flip a bit of $\hat{\be}$ equal to zero, and matching the one in the same position
of $\be$, causing a discrepancy. Because of Assumption \ref{ass:Pf_Pu}, the probability of such a transition is 
$\Pfzero{t+i}$. 
, while the probability of the self-loop transition from $\hat{t}_0=i$ to $\hat{t}_0=i$ itself
is $\Pmzero{t+i}$.

Note that, during the inner loop iterations of the RIP decoder acting on positions
of $\hat{\be}$ which have no discrepancies it is not possible to decrease the
value $\hat{t}_0$, as no reduction on the number of discrepancies between $\hat{\be}$
and $\be$ can be done changing values of $\hat{\be}$ which are already equal
to the ones in $\be$. Hence, we have that the probability of transitioning 
from $\hat{t}_0=i$ to $\hat{t}_0=i-1$ is zero.

The evolution of a PFSA can be computed simply taking the current state, 
represented as the vector $\val$ of probabilities for each FSA state and multiplying
it by an appropriate matrix which characterizes the transitions in the PFSA.
Such a matrix is derived as the adjacency matrix of the PFSA graph representation, 
keeping only the edges for which the read character matches the edge label, 
and substituting the one-values in the adjacency matrix with the probability
labelling the corresponding edge.
\begin{figure}[t]
    \centering
    \resizebox{\linewidth}{!}{
\begin{tikzpicture}[->, >=stealth', auto, semithick, node distance=4cm]

\tikzstyle{every state}=[fill=white,draw=black,thick,text=black,scale=1, inner sep=0pt, minimum size=1.5cm]
\tikzstyle{exit_state}=[fill=white,draw=myred,thick,text=myred,scale=1]
\node[state]    (A)                     {$t_1 = 0$};

\node[state]    (B)[right of=A]   {$t_1 = 1$};
\node[minimum size=1.5cm]         (C)[right of=B]   {$\ldots$};
\node[state]    (D)[right of=C]   {$t_1 = t$};

\path (B)   edge  node[below]{\texttt{1}, \textcolor{red}{$\Pfone{t^*-t+1}$}}     (A);
\path (C)   edge  node[below]{\texttt{1}, \textcolor{red}{$\Pfone{t^*-t+2}$}}   (B);
\path (D)   edge  node[below]{\texttt{1}, \textcolor{red}{$\Pfone{t^*}$}} (C);

\path (A)    edge [loop above] node{\texttt{1}, \textcolor{red}{$\Pmone{t^*-t}$}}     (A);
\path (B)    edge [loop above] node{\texttt{1}, \textcolor{red}{$\Pmone{t^*-t+1}$}}     (B);
\path (D)    edge [loop above] node{\texttt{1}, \textcolor{red}{$\Pmone{t^*}$}}     (D);
\end{tikzpicture}
 }
    \caption{Structure of the probabilistic FSA modeling the evolution of the 
    distribution of the $\hat{t}_1$ variable. Read characters are reported in black, 
    transition probabilites in red-}
    \label{fig:t1_evolution}
\end{figure}
We obtain the transition matrix modeling an iteration of the RIP decoder
acting on an $\hat{e}_j=e_j$ (i.e. reading a \texttt{0}) as the $(n-t+1)\times(n-t+1)$ matrix:
\begin{equation*}
    \bK_0 = \begin{bmatrix}
    \Pu{t}   & \cPu{t}  &  0        & 0       & 0        & 0\\
    0        & \Pu{t+1} & \cPu{t+1} & 0       & 0        & 0\\
    \vdots   & \vdots   & \vdots    & \vdots  & \vdots   & \vdots\\
    0        &       0  &       0   &      0  & \Pu{n-1} & \cPu{n-1} \\
    0        &       0  &       0   &      0  & 0        & \Pu{n}
    \end{bmatrix}
\end{equation*}

Since we want to compute the effect on the distribution of $\hat{t}_0$ after $n-t$ 
iterations of the RIP decoder acting on positions $j$ such that $\hat{e}_j=e_j$, 
we can obtain it simply as $\val\bK_0^{n-t}$. Note that the subsequent $t$
iterations of the RIP decoder will not alter the value of $\hat{t}_0$ as they act
on positions $j$ such that $e_j=1$.
Since we know that, at the beginning of the first iteration 
$\val=[\pr{\hat{t}_0=0}=1,\pr{\hat{t}_0=1}=0,\pr{\hat{t}_0=2}=0,\cdots,\pr{\hat{t}_0=n-t}=0]$, we are
able to compute $\prJ{\omega\xrightarrow{E_0}x}$ as the $(x+1)$-th element
of $\val\bK_0^{n-t}$.

We now model the distribution of $\hat{t}_1$, during the last $t$ iterations
of the inner loop of the RIP decoder performed during an iteration of the outer
loop. Note that, to this end, the first $n-t$ iterations of the inner loop
have no effect on $\hat{t}_1$. Denote with $t^*$ the incorrectly estimated bits 
$w_H(\be + \hat{\be})$ at the beginning of the inner loop iterations acting
on positions $j$ where $\hat{e}_j\neq e_j$. Note that, at the first iteration
of the outer loop of the RIP decoder, $t^*=\hat{t}_0+t$, when the RIP decoder
is about to analyze the first position for which $w_H(\be + \bar{\be})$.
Arguments analogous to the ones employed to model the PFSA describing
the evolution for $\hat{t}_0$ allow us to obtain the one modeling the evolution
for $\hat{t}_1$, reported in Figure~\ref{fig:t1_evolution}.

We are thus able to obtain the $\prJ{\omega\xrightarrow{E_1}x}$ PFSA reported in 
Figure~\ref{fig:t1_evolution} for $\hat{t}_1$ is 
$\textbf{z}=[\pr{\hat{t}_1=0}=0,\pr{\hat{t}_1=0}=0,\ldots,\pr{\hat{t}_1=t}=1]$ and employing the 
$(t+1)\times(t+1)$ transition matrix $\bK_1$ of the PFSA to compute $\textbf{z}\bK_1^t$.
The value of $\prJ{\omega\xrightarrow{E_1}x}$ corresponds to the $(x+1)$-th element
of $\textbf{z}\bK_1^t$.
 \section{Solving the Counting Subset Sum Problem}
\label{sec:appendix_b}
\begin{algorithm}[!t]
{ \scriptsize
\LinesNumbered
\DontPrintSemicolon
\caption{Computation of $\mathcal{N}_i(\val, \eta, \mathtt{thr})$ \label{alg:error_sets}}
\KwIn{ $\val$: an integer sequence, with elements in $\{0,\ldots v\}$, $|\val|=n$.
The collection admits repeated items\newline
$\eta$: the number of elements of the sought subsets of $\val$ \newline
$\mathtt{thr}$: the maximum allowed value of the sum of the $\eta$-wide integer subsets of $\val$ \newline
$i$: the number of distinct elements admitted in the subsets}
\KwOut{$\mathcal{N}_i(\val, \eta, \mathtt{thr})$: the number of subsets of $\val$, of $\eta$ integers picked 
with sum $\leq \mathtt{thr}$}
\KwData{$z$: the number of distinct elements in $\val$\newline
$\epsilon_i$: the $i$-th distinct integer in $\val$, $i \in \{0,\ldots,z-1\}$, $i < j \Rightarrow \epsilon_i < \epsilon_j$ \newline
$\lambda_i$: the number of occurrences (multiplicity) of $\epsilon_i$ in $\val$ \newline
}
$\mathtt{sum} \leftarrow 0$\;
 \If{$i = 1$}{
  \For{$j\leftarrow 0$ \KwTo $z-1$}{
     \tcp{Pick $\eta$ terms equal to $\epsilon_j$: their sum should be $\leq \mathtt{thr}$}
     \If{$(\epsilon_j \cdot \eta  \leq \mathtt{thr})\wedge (\lambda_j \geq \eta)$}{
          $\mathtt{sum} \leftarrow \mathtt{sum} + \binom{\lambda_i}{\eta}$\;
     }
  }
  \KwRet $\mathtt{sum}$\;
 } \Else {
  \For{$j \leftarrow 0$ \KwTo $z-1$}{
$m \leftarrow \min\{\lambda_j, \lfloor\frac{\mathtt{thr}}{\epsilon_j}\rfloor,\eta-(i-1) \}$\; 
    \tcp{$i-1$ \emph{distinct} terms must still be placed: place at most $\eta-(i-1)$}
    \For{$k \leftarrow 1$ to $m$}{
        $\mathtt{sum}\leftarrow \mathtt{sum} + \binom{\lambda_j }{k} \mathcal{N}_{(i-1)}\left(\val \setminus \{\epsilon_0 \ldots \epsilon_j\} , \eta-k, \mathtt{thr}-(k\cdot \epsilon_j)\right)$
    }
  }
 }
 \KwRet $\mathtt{sum}$ \;
}
\end{algorithm}

In the following, we describe the algorithm computing $\mathcal{N}(\val, \eta, \mathtt{thr})$, 
i.e., the number of subsets of the elements of $\val$, which have cardinality equal
to $\eta$, and which have the sum of their elements lesser than or equal to $\mathtt{thr}$.

In doing this, we leverage the fact that $\val$ has only a small number of 
distinct elements, $z \ll n=|\val|$. To this end, we represent $\val$ as the 
sequence of its $z$ distinct elements $[\epsilon_0, \epsilon_1,\ldots,
\epsilon_{z-1}]$ in increasing order of their value, i.e.,  $\forall i < j, e_i < e_j$.
Such a sequence is paired with the  sequence of the number of times that 
each $\epsilon_i$ appears in $\val$, $[\lambda_0, \lambda_1,\ldots,\lambda_{z-1}]$.

First of all, we note that the sets which are counted in 
$\mathcal{N}(\val, \eta, \mathtt{thr})$, can be partitioned according 
to the number of distinct elements contained in them. Denote with 
$\mathcal{N}_i(\val, \eta, \mathtt{thr})$ the number of 
the number of subsets of the elements of $\val$, with cardinality equal to $\eta$, 
sum lesser  or equal to $\mathtt{thr}$, and exactly $i$ distinct elements.
The value of $\mathcal{N}(\val, \eta, \mathtt{thr})$ is obtained as the sum 
over all  $i\in{1,\ldots,z}$ of the values of 
$\mathcal{N}_i(\val, \eta, \mathtt{thr})$. The computation
of $\mathcal{N}_i(\val, \eta, \mathtt{thr})$ is described in Algorithm~\ref{alg:error_sets}.

 \section{Proof of Lemma 1}
\label{sec:appendix_d}
\begin{lemma}
From Assumption~\ref{ass:row_probability}, the probabilities that the $i$-th bit of the syndrome $(0\leq i\leq r-1)$ is asserted knowing that the $z$-th bit of the error vector $(0\leq z\leq n-1)$ is null or not, i.e., $\pr{s_i = 1 | e_z} = \pr{\langle \bh_{i,:},\be \rangle = 1 | e_z}$, $\langle \bh_{i,:},\be \rangle  = \bigoplus_{j=0}^{n-1} h_{i,j} \cdot e_j$, can be expressed for each bit position $z$, $0\leq z \leq n-1$, of the error vector as follows:
$$
\pErrZeroUnsat=\pr{\langle \bh_{i,:},\be \rangle = 1\ |\ e_z=0} = \frac{\sum_{l=0,\text{ l odd}}^{\min\{w,t\}}\binom{w}{l}\binom{n-w}{t-l}}{\binom{n-1}{t}}
$$
$$
\pErrOneUnsat=\pr{\langle \bh_{i,:},\be \rangle = 1\ |\ e_z=1} = \frac{\sum_{l=0,\text{ l even}}^{\min\{w-1,t-1\}}\binom{w-1}{l}\binom{n-w}{t-1-l}}{\binom{n-1}{t-1}}
$$
Consequentially, the probability that Algorithm~\ref{alg:RIP_BF} performs a bit-flip of an element of the estimated error vector, $\hat{e}_z$, when the corresponding bit of the actual error vector is asserted, $e_z=1$, i.e., $\mathrm{P}_{f\mid1}$, and the  probability that Algorithm~\ref{alg:RIP_BF} maintains the value of the estimated error vector, $\hat{e}_z$, when the corresponding bit of the actual error vector is null, $e_z=0$, i.e., $\mathrm{P}_{m\mid0}$, are:
$$\mathrm{P}_{f\mid1}= \sum_{\upc = b}^{v}\binom{v}{\upc}\pic^{\upc}(1-\pic)^{v-\upc},$$
$$\mathrm{P}_{m\mid0} = \sum_{\upc = 0}^{b-1}\binom{v}{\upc}\pci^{\upc}(1-\pci)^{v-\upc}.$$
\end{lemma}
\begin{proof}
For the sake of brevity, we consider the case of $e_z=1$ deriving the expression of $\mathrm{P}_{f\mid1}$; the proof for $\mathrm{P}_{m\mid0}$ can be carried out with similar arguments.
Given a row $\bh_{i,:}$ of the parity-check matrix $\bH$, such that $z \in \supp{\bh_{i,:}}$, the equation $\bigoplus_{j=0}^{n-1} h_{i,j} \cdot e_j$ (in the unknown $\be$) yields a non-null value for the $i$-th bit of the syndrome, $s_i$, (i.e., the eq. is unsatisfied) if and only if the support of the error vector $\be$ is such that $\bigoplus_{j=0}^{n-1} h_{i,j} \cdot e_j= 2a+1, a\geq 0$, including the term having $j=z$, i.e., $\bh_{i,z} \cdot e_z=1$. 
This implies that the cardinality of the set obtained intersecting the support $\bh_{i,:}$ with the one of $\be$, $\card{(\supp{\bh_{i,:}}\setminus \{z\})\cap(\supp{\be}\setminus \{z\})}$, must be an even number, which in turn cannot be larger than the minimum between $\card{\supp{\bh_{i,:}}\setminus \{i\}} = w-1$ and $\card{\supp{\be}\setminus \{i\}} = t-1$. 

The probability $\pic$ is obtained considering the fraction of the number of error vector values having an even number of asserted bits matching the asserted bits ones in a row 
of $\bH$ (noting that, for the $z$-th bit position, both the error and the row of $\bH$ are set) on the number of error vector values having $t-1$ asserted bits over $n-1$ positions, i.e., $\binom{n-1}{t-1}$. 
The numerator of the said fraction is easily computed as the sum of all error vector configurations having an even number $0 \leq l \leq \mathrm{min}\{w-1,t-1\}$ of asserted bits. Considering a given value for $l$, the counting of the error vector values is derived as follows. 
Picking one of vector with $l$ asserted bits over $w$ possible positions, i.e., one vector over $\binom{w-1}{l}$ possible ones, there are $\binom{n-w}{t-1-l}$ possible values of the error vector exhibiting $t-1-l$ null bits in the remaining $n-w$ positions; therefore, the total number of vectors with weigh $l$ is $\binom{w-1}{l} \cdot \binom{n-w}{t-1-l}$. Repeating the same line of reasoning for each value of $l$ allows to derive the numerator of the formula defining $\pic$.   

From Assumption~\ref{ass:row_probability}, the value of any row $\bh_{i,:}$ is modeled as a random variable with a Bernoulli distribution having parameter (or expected value) $\pic$, and each of these random variables is independent from the others. 
Consequentially, the probability that Algorithm~\ref{alg:RIP_BF} performs a bit-flip of an element of the estimated error vector when the corresponding bit of the actual error vector is asserted and the counter of the unsatisfied parity checks (upc) is above or equal to a given threshold $b$, is derived as the binomial probability obtained adding the outcomes of $v$ (column-weight of $\bH$) i.i.d. Bernoulli trials.
\qed
\end{proof}
 \section{Proof of Lemma 2}
\label{sec:appendix_c}
\textbf{Lemma 2.} The execution path of the inner loop in 
Algorithm~\ref{alg:RIP_BF} yielding the worst possible decoder success rate is
the one taking place when $\pi^* \in \mathcal{P}^*_n$ is applied at the
beginning of the outer loop, that is:
$$\forall \pi\in\mathcal{P}_n, \forall \pi^*\in\mathcal{P}_n^*, \ \ \pr{\left.\hat{\be} \neq \be\right| \hspace{1mm}\pi\in\mathcal{P}_n}
\leq \pr{\left.\hat{\be} \neq \be\right| \hspace{1mm}\pi^*\in\mathcal{P}^*_n}.$$
\begin{proof}
First of all, we can write $\pr{\left.\be'\neq \be\right| \hspace{1mm}\pi\in\mathcal{P}_n} = 1-\beta(\pi),$
where $\beta(\pi)$ is the probability that all bits, evaluated in the order specified by $\pi$, are correctly processed. 
To visualize the effect of a permutation $\pi^*\in\mathcal{P}_n$, we can consider the following representation
\begin{equation*}
\pi^*(\be)\oplus \pi^*(\bar{\be})  = [\underbrace{0,0,\cdots,0}_{\text{length $n-\hat{t}$}},\underbrace{1,1,\cdots,1}_{\text{length $\hat{t}$}}], \hspace{2mm}\forall \pi^*\in\mathcal{P}^*_n.
\end{equation*}
The decoder will hence analyze first a run of $n-\hat{t}$ positions where the 
differences between the permuted error $\pi^*(\be)$ vector and $\pi^*(\bar{\be})$ 
contain only zeroes, followed by a run of  $\hat{t}$ positions containing only ones.
Thus, we have that
$$\beta(\pi^*) =   \left(\Pmzero{\hat{t}}\right)^{n-\hat{t}}\cdot \Pfone{\hat{t}}
\cdot \Pfone{\hat{t}-1} \cdots \Pfone{1}$$
The former expression can be derived thanks to Assumption \ref{ass:Pf_Pu} as follows.
Note that, the first elements in the first $n-\hat{t}$ positions of $\pi^*(\hat{\be})$ 
and $\pi^*(\be)$ match, therefore the decoder makes a correct evaluation if
it does not change the value of $\pi^*(\hat{\be})$. This in turn implies that, 
in case a sequence of $n-\hat{t}$ correct decisions are made in the corresponding
iterations of the inner loop, each iteration will have the same probability 
$\Pmzero{\hat{t}}$ correctly evaluating the current estimated error bit.
This leads to a probability of performing the first $n-\hat{t}$ iterations
taking a correct decision equal to $\left(\Pmzero{\hat{t}}\right)^{n-\hat{t}}$
Through an analogous line of reasoning, observe that the decoder will need to 
change the value of the current estimated error bit during the last $\hat{t}$
iterations of the inner loop. As a consequence, if all correct decisions are made, 
the number of residual errors will decrease by one at each inner loop iteration, yielding
the remaining part of the expression.

Consider now a generic permutation $\pi$, such that the resulting $\pi(\be)$ 
has support $\{u_0,\cdots,u_{\hat{t}-1}\}$; we have
{\small
\begin{align*}
\beta(\pi) & \nonumber = \left[\Pu{\hat{t}}\right]^{u_0} \Pf{t} \left[\Pu{\hat{t}-1}\right]^{u_1-u_0-1}\Pf{\hat{t}-1}\cdots \Pf{1} \left[\Pu{0}\right]^{n-1 - u_{\hat{t}-1}}\\
& = \left[\Pu{t}\right]^{u_0}\left[\Pu{0}\right]^{n-1 - u_{\hat{t}-1}}\prod_{j = 1}^{\hat{t}-1}\left[\Pu{\hat{t}-j}\right]^{u_{j}-u_{j-1}-1}\prod_{l = 0}^{\hat{t}-1}\Pf{\hat{t}-l}.
\end{align*}
}
We now show that we always have $\beta(\pi)\geq \beta(\pi^*)$.
Indeed, since $\Pu{0}=1$ and due to the monotonic trends of $\mathrm{P}_u$ and $\mathrm{P}_f$, the following chain of inequalities can be derived
{\small
\begin{align*}
\beta(\pi)&\nonumber =\left[\Pu{0}\right]^{n-1 - u_{\hat{t}-1}}\left[\Pu{\hat{t}}\right]^{u_0}\prod_{j = 1}^{\hat{t}-1}\left[\Pu{\hat{t}-j}\right]^{u_{j}-u_{j-1}-1}\prod_{l = 0}^{\hat{t}-1}\Pf{\hat{t}-l}\\& \geq \left[\Pu{0}\right]^{n-1 - u_{\hat{t}-1}}\left[\Pu{\hat{t}}\right]^{u_0}\prod_{j = 1}^{\hat{t}-1}\left[\Pu{\hat{t}}\right]^{u_{j}-u_{j-1}-1}\prod_{l = 0}^{\hat{t}-1}\Pf{\hat{t}-l}\\
& = \left[\Pu{0}\right]^{n-1 - u_{\hat{t}-1}}\left[\Pu{\hat{t}}\right]^{u_0}\left[\Pu{\hat{t}}\right]^{u_{\hat{t}-1}-u_{0}-(\hat{t}-1)}\prod_{l = 0}^{\hat{t}-1}\Pf{\hat{t}-l})\\
& = \left[\Pu{0}\right]^{n-1 - u_{\hat{t}-1}}\left[\Pu{\hat{t}}\right]^{u_{\hat{t}-1}-(\hat{t}-1)}\prod_{l = 0}^{\hat{t}-1}\Pf{\hat{t}-l}\\
& \geq \left[\Pu{\hat{t}}\right]^{n-1 - u_{\hat{t}-1}}
\left[\Pu{\hat{t}}\right]^{u_{\hat{t}-1}-(\hat{t}-1)}\prod_{l = 0}^{\hat{t}-1}\Pf{\hat{t}-l}\\
& = \left[\Pu{\hat{t}}\right]^{n-\hat{t}}\prod_{l = 0}^{\hat{t}-1}\Pf{\hat{t}-l} = \beta(\pi^*).
\end{align*}}\qed
\end{proof}
 
\end{document}